\documentclass[runningheads]{llncs}

\usepackage{hyperref}
\usepackage{xcolor}
\usepackage{graphicx}
\usepackage{eso-pic}

\usepackage{array, makecell} 
\usepackage{multirow}
\usepackage{epstopdf}
\usepackage{ulem}
\usepackage{caption}
\usepackage{footnote}
\usepackage{booktabs}
\usepackage{amsmath}
\usepackage{amsfonts}
\usepackage{amssymb}
\usepackage{upgreek}
\usepackage{thmtools, thm-restate}

\usepackage{algorithm}
\usepackage[noend]{algpseudocode}

\title{An Automated Market Maker Minimizing Loss-Versus-Rebalancing}

\usepackage{authblk}
\author{Conor McMenamin\inst{1} \and Vanesa Daza\inst{1,2} \and Bruno Mazorra\inst{1}}

\authorrunning{McMenamin, Daza and Mazorra}

\institute{Department of Information and Communication Technologies, Universitat Pompeu Fabra, Barcelona, Spain \and CYBERCAT - Center for Cybersecurity Research of Catalonia}

\usepackage{cancel}

\usepackage{pifont}

\def\bitcoin{%
  \leavevmode
  \vtop{\offinterlineskip 
    \setbox0=\hbox{B}%
    \setbox2=\hbox to\wd0{\hfil\hskip-.03em
    \vrule height .3ex width .15ex\hskip .08em
    \vrule height .3ex width .15ex\hfil}
    \vbox{\copy2\box0}\box2}}

\newcommand{\MIFP}{\varepsilon}

\newcommand{\pool}{\Phi}

\newcommand{\newconstruct}[5]{%
  \newenvironment{ALC@\string#1}{\begin{ALC@g}}{\end{ALC@g}}
   \newcommand{#1}[2][default]{\ALC@it#2\ ##2\ #3%
     \ALC@com{##1}\begin{ALC@\string#1}}
   \ifthenelse{\boolean{ALC@noend}}{
     \newcommand{#4}{\end{ALC@\string#1}}
   }{
     \newcommand{#4}{\end{ALC@\string#1}\ALC@it#5}
   } 
}

\algdef{SE}[EVENT]{Upon}{EndUpon}[1]{\textbf{upon}\ #1\ \algorithmicdo}{\algorithmicend}\algtext*{EndUpon}

\begin{document}

\maketitle

\begin{abstract}
  The always-available liquidity of automated market makers (AMMs) has been one of the most important catalysts in early cryptocurrency adoption. However, it has become increasingly evident that AMMs in their current form are not viable investment options for passive liquidity providers. This is large part due to the cost incurred by AMMs providing stale prices to arbitrageurs against external market prices, formalized as loss-versus-rebalancing (LVR) [Milionis et al., 2022].

  In this paper, we present Diamond, an automated market making protocol that aligns the incentives of liquidity providers and block producers in the protocol-level retention of LVR. In Diamond, block producers effectively auction the right to capture any arbitrage that exists between the external market price of a Diamond pool, and the price of the pool itself. The proceeds of these auctions are shared by the Diamond pool and block producer in a way that is proven to remain incentive compatible for the block producer. Given the participation of competing arbitrageurs to capture LVR, LVR is minimized in Diamond. 
  We formally prove this result, and detail an implementation of Diamond. We also provide comparative simulations of Diamond to relevant benchmarks, further evidencing the LVR-protection capabilities of Diamond.
  With this new protection, passive liquidity provision on blockchains can become rationally viable, beckoning a new age for decentralized finance.
\end{abstract}
\section{Introduction}

CFMMs such as Uniswap \cite{UniswapWebsite} have emerged as the dominant class of AMM protocols. 
CFMMs offer several key advantages for decentralized liquidity provision. They are efficient computationally, have minimal storage needs, matching computations can be done quickly, and liquidity providers can be passive. Thus, CFMMs are uniquely suited to the severely computation- and storage-constrained environment of blockchains. 

Unfortunately, the benefits of CFMMs are not without significant costs. One of these costs is definitively formalized in \cite{LVRRoughgarden} as \textit{loss-versus-rebalancing} (LVR). It is proved that as the underlying price of a swap moves around in real-time, the discrete-time progression of AMMs leave arbitrage opportunities against the AMM. In centralized finance, market makers typically adjust to new price information before trading. This comes at a considerable cost to AMMs (for CFMMs, \cite{LVRRoughgarden} derives the cost to be quadratic in realized moves), with similar costs for AMMs derived quantitatively in \cite{Park2021TheCF,DEXCritiqueCapponi}, and presented in Figure \ref{fig:Toxic}.

\begin{figure}[t]
    \centering
    \includegraphics[scale=0.45]{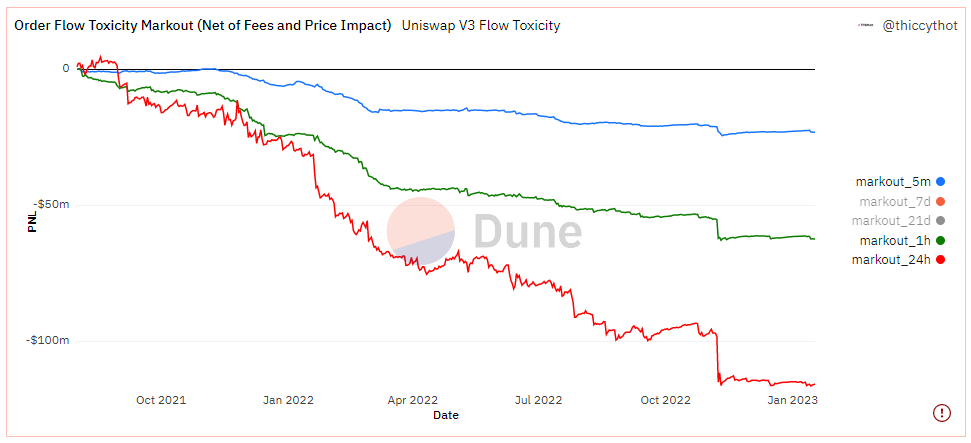}
    \caption[Toxicity of Uniswap V3 Order Flow \cite{duneToxicityQuery}.]{Toxicity of Uniswap V3 Order Flow \cite{duneToxicityQuery}. This graph aggregates the PnL (\textit{toxicity}) of all trades on the Uniswap V3 WETH/USDC pool, measuring PnL of each order after 5 minutes, 1 hour, and 1 day. These are typical time periods within which arbitrageurs close their positions against external markets. This demonstrates the losses being incurred in existing state-of-the-art DEX protocols are significant, consistent, and unsustainable; toxic.}
    \label{fig:Toxic}
\end{figure}

These costs are being realized by liquidity providers in current AMM protocols. 
All of these factors point towards unsatisfactory protocol design, and a dire need for an LVR-resistant automated market maker. In this paper, we provide Diamond, an AMM protocol which formally protects against LVR.

\subsection{Our Contribution}

We present Diamond, an AMM protocol which isolates the LVR being captured from a Diamond liquidity pool, and forces some percentage of these LVR proceeds to be returned to the pool. As in typical CFMMs, Diamond pools are defined with respect to two tokens $x$ and $y$. At any given time, the pool has reserves of $R_x$ and $R_y$ of both tokens, and some pool pricing function\footnote{See Equation \ref{eq:pricerestrictions} for a full description of pool pricing functions as used in this paper} $\textit{PPF}(R_x,R_y)$. We demonstrate our results using the well-studied Uniswap V2 pricing function of $\textit{PPF}(R_x,R_y)=\frac{R_x}{R_y}$. 

In Diamond, block producers are able to capture the block LVR of a Diamond pool themselves or auction this right among a set of competing arbitrageurs. In both cases, the block producer revenue approximates the arbitrage revenue. Therefore, block producers are not treated as traditional arbitrageurs, but rather as players with effective arbitrage capabilities due to their unique position in blockchain protocols. 

For each Diamond pool, we introduce the concept of its \textit{corresponding CFMM pool}. Given a Diamond pool with token reserves $(R_x,R_y)$ and pricing function $\textit{PPF}(R_x,R_y)=\frac{R_x}{R_y}$, the corresponding CFMM pool is the Uniswap V2 pool with reserves $(R_x,R_y)$. If a block producer tries to move the price of the corresponding CFMM pool adding $\Upsilon_x$ tokens and removing $\Upsilon_y$, the same price is achieved in the Diamond pool by adding $(1-\beta)\Upsilon_x$ tokens for some $\beta>0$, with $\beta$ the \textit{LVR rebate parameter}. The block producer receives $(1-\beta)\Upsilon_y$. In our framework, it can be seen that $\textit{PPF}(R_x+(1-\beta)\Upsilon_x, R_y-(1-\beta)\Upsilon_y)<\textit{PPF}(R_x+\Upsilon_x,R_y-\Upsilon_y)$, which also holds in our example, as $\frac{R_x+(1-\beta)\Upsilon_x}{R_y-(1-\beta)\Upsilon_y}<\frac{R_x+\Upsilon_x}{R_y-\Upsilon_y}$. A further $\upsilon_y$ tokens are removed from the Diamond pool to move the reserves to the same price as the corresponding CFMM pool, with these tokens added to a \textit{vault}.

Half of the tokens in the vault are then periodically converted into the other token (at any time, all tokens in the vault are of the same denomination) in one of the following ways:
\begin{enumerate}
    \item An auction amongst arbitrageurs.
    \item Converted every block by the block producer at the final pool price. If the block producer must buy $\frac{\eta}{2}$ tokens to convert the vault, the block producer must simultaneously sell $\frac{\eta}{2}$ futures which replicate the price of the token to the pool. These futures are then settled periodically, either by \begin{enumerate}
        \item Auctioning the $\frac{\eta}{2}$ tokens corresponding to the futures amongst competing arbitrageurs with the protocol paying/collecting the difference.
        \item The use of a decentralized price oracle. In this paper, we consider the use of the settlement price of an on-chain frequent batch auction, such as that of \cite{FairTraDEXMcMenamin}, which is proven to settle at the external market price in expectancy.
    \end{enumerate}
\end{enumerate}

Importantly, these auctions are not required for protocol liveness, and can be arbitrarily slow to settle. We prove that all of these conversion processes have 0 expectancy for the block producer or Diamond pool, and prove that the LVR of a Diamond pool is $(1-\beta)$ of the corresponding CFMM pool. Our implementation of Diamond isolates LVR arbitrageurs from normal users, using the fact that arbitrageurs are always bidding to capture LVR. Specifically, if an LVR opportunity exists at the start of the block, an arbitrageur will bid for it in addition to ordering normal user transactions, meaning the proceeds of a block producer are at least the realized LVR, with LVR corresponding to the difference between the start- and end-states of an AMM in a given block.
This ensures the protections of Diamond can be provided in practice while providing at least the same trading experience for normal users. Non-arbitrageur orders in a Diamond pool can be performed identically to orders in the corresponding CFMM pool after an arbitrageur has accepted to interact with the pool through a special arbitrageur-only transaction. Although this means user orders may remain exposed to the front-running, back-running and sandwich attacks of corresponding CFMMs, the LVR retention of Diamond pools should result in improved liquidity and reduced fees for users. 

We discuss practical considerations for implementing Diamond, including decreasing the LVR rebate parameter, potentially to 0, during periods of protocol inactivity until transactions are processed, after which the parameter should be reset. This ensures the protocol continues to process user transactions, which becomes necessary when arbitrageurs are not actively extracting LVR. If arbitrageurs are not arbitraging the pool for even small LVR rebate parameters, it makes sense to allow transactions to be processed as if no LVR was possible. In this case, Diamond pools perform identically to corresponding CFMM pools. However, if/when arbitrageurs begin to compete for LVR, we expect LVR rebate parameters to remain high.

We present a series of experiments in Section \ref{sec:Diamond:sims} which isolate the benefits of Diamond. We compare a Diamond pool to its corresponding Uniswap V2 pool, as well as the strategy of holding the starting reserves of both tokens, demonstrating the power of Diamond. We isolate the effects of price volatility, LVR rebate parameter, pool fees, and pool duration on a Diamond pool. Our experiments provide convincing evidence that the relative value of a Diamond pool to its corresponding Uniswap V2 pool is increasing in each of these variables. These experiments further evidence the limitations of current CFMMs, and the potential of Diamond.

\subsection{Organization of the Paper}

Section \ref{sec:Diamond:RW} analyzes previous work related to LVR in AMMs. Section \ref{sec:Diamond:Prelims} outlines the terminology used in the paper. Section \ref{sec;protocol} introduces the Diamond protocol. Section \ref{sec:Diamond:properties} proves the properties of Diamond. Section \ref{sec:Diamond:implementation} describes how to implement the Diamond protocol, and practical considerations which should be made. Section \ref{sec:Diamond:sims} provides an analysis Diamond over multiple scenarios and parameters, including a comparison to various reference strategies. We conclude in Section \ref{sec:Diamond:conclusion}.

\section{Related Work}\label{sec:Diamond:RW}

There are many papers on the theory and design of AMMs, with some of the most important including \cite{Adams2020UniswapVC,Adams2021UniswapVC,LVRRoughgarden,AMMTheoryBartoletti,AMMMEVBartoletti}. The only peer-reviewed AMM design claiming protection against LVR \cite{DynamicAMMsKrishnamachari} is based on live price oracles. The AMM must receive the price of a swap before users can interact with the pool. Such sub-block time price data requires centralized sources which are prone to manipulation, or require the active participation of AMM representatives, a contradiction of the passive nature of AMMs and their liquidity providers. We see this as an unsatisfactory dependency for DeFi protocols. 

Attempts to provide LVR protection without explicit use of oracles either use predictive fees for all players \cite{OptimalFeesChitra} and/or reduce liquidity for all players through more complex constant functions \cite{LVRSDAMMBichuch}. Charging all users higher fees to compensate for arbitrageur profits reduces the utility of the protocol for genuine users, as does a generalized liquidity reduction. In Diamond, we only reduce liquidity for arbitrageurs (which can also be seen as an increased arbitrageur-specific fee), providing at least the same user experience for typical users as existing AMMs without LVR protection.

A recent proposed solution to LVR published in a blog-post \cite{McAMMs} termed \textit{MEV-capturing AMMs} (McAMMs) considers auctioning off the first transaction/series of transaction in an AMM among arbitrageurs, with auction revenue paid in some form to the protocol. Two important benefits of Diamond compared to the proposed McAMMs are the capturing of realized LVR in Diamond as opposed to predicted LVR in McAMMs, and decentralized access to Diamond compared to a single point of failure in McAMMs.

In McAMMs, bidders are required to predict upcoming movements in the AMM. Bidders with large orders to execute over the period (e.g. private price information, private order flow, etc.) have informational advantages over other bidders. Knowing the difference between expected LVR excluding this private information vs. true expected LVR allows the bidder to inflict more LVR on the AMM than is paid for. As this results in better execution for the winner's orders, this may result in more private order flow, which exacerbates this effect. Diamond extracts a constant percentage of the true LVR, regardless of private information.
McAMMs also centralize (first) access control to the winning bidder. If this bidder fails to respond or is censored, user access to the protocol is prohibited/more expensive. Diamond is fully decentralized, incentive compatible and can be programmed to effectively remove LVR in expectancy. Future McAMM design improvements based on sub-block time auctions are upper-bounded by the current protection provided by Diamond.

\section{Preliminaries}\label{sec:Diamond:Prelims}

This section introduces the key terminology and definitions needed to understand LVR, the Diamond protocol, and the proceeding analysis. In this work we are concerned with a single swap between token $x$ and token $y$. We use $x$ and $y$ subscripts when referring to quantities of the respective tokens. The external market price of a swap is denoted by $\MIFP$, while pool prices and price functions are denoted using a lowercase $p$ and uppercase $P$ respectively. The price of a swap is quoted as the quantity of token $x$ per token $y$. 

In this work we treat the block producer and an arbitrageur paying for the right to execute transactions in a block as the same entity. This is because the the arbitrageur must have full block producer capabilities, and vice versa, with the payoff for the block producer equal to that of an arbitrageur under arbitrageur competition. For consistency, and to emphasize the arbitrage that is taking place in extracting LVR, we predominantly use the arbitrageur naming convention. That being said, it is important to remember that this arbitrageur has exclusive access to building the sub-block of Diamond transactions. Where necessary, we reiterate that it is the block producer who control the per-block set of Diamond transactions, and as such, the state of the Diamond protocol.

\subsection{Constant Function Market Makers}\label{sec:Diamond:CFMMs}

A CFMM is characterized by \textit{reserves} $(R_x,R_y) \in \mathbb{R}_{+}^2$ which describes the total amount of each token in the pool. The price of the pool is given by \textit{pool price function} $PPF: \mathbb{R}^2_{+} \rightarrow \mathbb{R}$ taking as input pool reserves $(R_x,R_y)$. $PPF$ has the following properties:
\begin{equation}\label{eq:pricerestrictions}
    \begin{aligned}
        \text{(a)}&\  PPF \text{ is everywhere differentiable, with }\frac{\partial PPF}{\partial R_x}>0,\ \frac{\partial PPF}{\partial R_y}<0. \\
        \text{(b)}& \ \lim_{R_x \rightarrow 0} PPF = 0, \ \lim_{R_x \rightarrow \infty} PPF = \infty, \ \lim_{R_y \rightarrow 0} PPF = \infty, \ \lim_{R_y \rightarrow \infty} PPF = 0. \\
        \text{(c)}& \ \text{If } \textit{PPF}(R_x,R_y)=p, \text{ then } \textit{PPF}(R_x+ cp, R_y + c )=p, \ \forall c >0. \\
    \end{aligned}
\end{equation}

These are typical properties of price functions. Property (a) states the price of $y$ is increasing in the number of $x$ tokens in the pool and decreasing in the number of $y$ tokens. Property (b) can be interpreted as any pool price value is reachable for a fixed $R_x$, by changing the reserves of $R_y$, and vice versa. Property (c) states that adding reserves to a pool in a ratio corresponding to the current price of the pool does not change the price of the pool. These properties trivially hold for the Uniswap V2 price function of $\frac{R_x}{R_y}$, and importantly allow us to generalize our results to a wider class of CFMMs.

For a CFMM, the \textit{feasible set of reserves} $C$ is described by:
\begin{equation}
    C = \{ (R_x,R_y) \in \mathbb{R}_{+}^2 : \textit{PIF}(R_x,R_y)=k \}
\end{equation}
where $\textit{PIF}:\mathbb{R}_{+}^2 \rightarrow \mathbb{R}$ is the pool invariant and $k \in \mathbb{R}$ is a constant. The pool is defined by a smart contract which allows any player to move the pool reserves from the current reserves $(R_{x,0},R_{y,0}) \in C$ to any other reserves $(R_{x,1},R_{y,1}) \in C$ if and only if the player provides the difference $(R_{x,1}- R_{x,0} ,R_{y ,1}-R_{y,0})$.

Whenever an arbitrageur interacts with the pool, say at time $t$ with reserves $(R_{x,t},R_{y,t})$, we assume as in \cite{LVRRoughgarden} that the arbitrageur maximizes their profits by exploiting the difference between $\textit{PPF}(R_{x,t},R_{y,t})$ and the external market price at time $t$, denoted $\MIFP_t$. To reason about this movement, we consider a \textit{pool value function} $V: \mathbb{R}_{+} \rightarrow \mathbb{R}$ defined by the optimization problem:

\begin{equation}\label{eq:poolOptimize}
    V(\MIFP_t) = \displaystyle\min\limits_{(R_x,R_y) \in \mathbb{R}_{+}^2}\MIFP_t R_y + R_x, \text{ such that } \textit{PIF}(R_x,R_y)=k
\end{equation}
Given an arbitrageur interacts with the pool with external market price $\MIFP_t$, the arbitrageur moves the pool reserves to the $(R_x,R_y)$ satisfying $V(\MIFP_t)$.

\subsection{Loss-Versus-Rebalancing}\label{sec:Diamond:LVR}

LVR, and its prevention in AMMs is the primary focus of this paper. The formalization of LVR \cite{LVRRoughgarden} has helped to illuminate one of the main costs of providing liquidity in CFMMs.
The authors of \cite{LVRRoughgarden} provide various synonyms to conceptualize LVR. In this paper, we use the opportunity cost of arbitraging the pool against the external market price of the swap, which is proven to be equivalent to LVR in Corollary 1 of \cite{LVRRoughgarden}. The LVR between two blocks $B_t$ and $B_{t+1}$ where the reserves of the AMM at the end of $B_t$ are $(R_{x,t},R_{y,t})$ and the external market price when creating block $B_{t+1}$ is $\MIFP_{t+1}$ is:
\begin{equation}\label{eq:LVR}
    R_{x,t}+R_{y,t}\MIFP_{t+1}- V(\MIFP_{t+1}) = (R_{x,t}-R_{x,t+1})+ (R_{y,t}-R_{y,t+1})\MIFP_{t+1}.
\end{equation}
As this is the amount being lost to arbitrageurs by the AMM, this is the quantity that needs to be minimized in order to provide LVR protection. In Diamond, this minimization is achieved. 

\subsection{Auctions}

To reason about the incentive compatibility of parts of our protocol, we outline some basic auction theory results.

\textbf{First-price-sealed-bid-auction}: There is a finite set of players $\mathcal I$ and a single object for sale. Each bidder $i\in\mathcal I$ assigns a value of $X_i$ to the object. Each $X_i$ is a random variable that is independent and identically distributed on some interval $[0,V_{max}]$. The bidders know its realization $x_i$ of $X_i$. We will assume that bidders are risk neutral, that they seek to maximize their expected payoff. Per auction, each player submit a bid $b_i$ to the auctioneer. The player with the highest bid gets the object and pays the amount bid. In case of tie, the winner of the auction is chosen randomly. Therefore, the utility of a player $i\in\mathcal I$ is
\begin{equation*}
    u_i(b_i,b_{-i})=\begin{cases}
                    \frac{x_i-b_i}{m},\text{ if }b_i=\text{max}_i\{b_i\},\\
                    0, \ \ \ \ \ \ \text{ otherwise}
                     \end{cases}
\end{equation*}
where $m=|\text{argmax}_i\{b_i\}|$. In our protocol, we have an amount of tokens $z$ that will be auctioned. This object can be exchanged by all players at the external market price $\MIFP$. In this scenario, we have the following lemma. Proofs are included in the Appendix

\begin{restatable}{lemma}{auction}\label{lem:incentiveCompatibleAuctions}
Let $\mathcal I$ be a set of players that can exchange at some market any amount of tokens $x$ or $y$ at the external market price $\MIFP$. If an amount $z$ of token  $y$ is auctioned in a first-price auction, then the maximum bid of any Nash equilibrium is at least $z\MIFP$.
\end{restatable}

\section{\texorpdfstring{Diamond}{TEXT}}\label{sec;protocol}

This section introduces the Diamond protocol. When the core protocol of Section \ref{sec;description} is run, some amount of tokens are removed from the pool and placed in a \textit{vault}. These vault tokens are eventually re-added to the pool through a conversion protocol. Sections \ref{sec:Diamond:conv1} and \ref{sec:Diamond:conv2} detail two conversion protocols which can be run in conjunction with the core Diamond protocol. Which conversion protocol to use depends on the priorities of the protocol users, with a  discussion of their trade-offs provided in Section \ref{sec:Diamond:sims}, and represented graphically in Figure \ref{fig:StrategyComparison}. 
These trade-offs can be summarized as follows: 
\begin{itemize}
    \item The process of Section \ref{sec:Diamond:conv1} forces the arbitrageur to immediately re-add the removed tokens to the Diamond pool, while ensuring the ratio of pool tokens equals the external market price. This ratio is achieved by simultaneously requiring the arbitrageur to engage in a futures contract tied to the pool price, with the arbitrageur taking the opposite side of the contract. These futures offset any incentive to manipulate the ratio of tokens. This results in a higher variance of portfolio value for both the Diamond pool and the arbitrageur. In return for this risk, this process ensures the pool liquidity is strictly increasing in expectancy every block, with the excess value (reduced LVR) retained by the vault immediately re-added to the pool. This process can be used in conjunction with a decentralized price oracle to ensure the only required participation of arbitrageurs is in arbitraging the pool (see process 2 in Section \ref{sec:Diamond:conv1}). It should be noted that these futures contracts have collateral requirements for the arbitrageur, which has additional opportunity costs for the arbitrageur. 
    \item The process in Section \ref{sec:Diamond:conv2} converts the vault tokens periodically. This can result in a large vault balance accruing between conversions, with this value taken from the pool. This means the quality (\textit{depth}) of liquidity is decreasing between conversions, increasing the impact of orders. From the AMM's perspective, this process incurs less variance in the total value of tokens owned by the pool (see Figure \ref{fig:StrategyComparison}), and involves a more straightforward and well-studied use of an auction (compared to a trusted decentralized oracle). There is also no collateral requirement for the arbitrageur outside of the block in which the arbitrage occurs. \footnote{As the arbitrageur and block producer are interchangeable from Diamond's perspective, we see the requirement for the block producer/arbitrageur to provide collateral in a block controlled by the block producer as having negligible cost.}
\end{itemize}

Section \ref{sec:Diamond:properties} formalizes the properties of Diamond, culminating in Theorem \ref{thm:main}, which states that Diamond can be parameterized to reduce LVR arbitrarily close to 0. It is important to note that Diamond is not a CFMM, but the rules for adjusting pool reserves are dependent on a CFMM.

\subsection{Model Assumptions}\label{sec:Diamond:assumptions}

We outline here the assumptions used when reasoning about Diamond. In keeping with the seminal analysis of \cite{LVRRoughgarden}, we borrow a subset of the assumptions therein, providing here a somewhat more generalized model. 

\begin{enumerate}
    \item External market prices follow a martingale process.
    \item The risk-free rate is 0.
    \item There exists a population of arbitrageurs able to frictionlessly trade at the external market price, who continuously monitor and periodically interact with AMM pools.
    \item An optimal solution $(R^*_x,R^*_y) $ to Equation \ref{eq:poolOptimize} exists for every external market price $\MIFP\geq 0$. 
\end{enumerate}

The use of futures contracts in one version of the Diamond protocol makes the  risk-free rate an important consideration for implementations of Diamond. If the risk free rate is not 0, the profit or loss related to owning token futures vs. physical tokens must be considered. Analysis of a non-zero risk-free rate is beyond the scope of the thesis.

\subsection{Core Protocol}\label{sec;description}

We now describe the core Diamond protocol, which is run by all Diamond variations. A Diamond pool $\Phi$ is described by reserves $(R_{x},R_{y})$, a pool pricing function $\textit{PPF}()$, a pool invariant function $\textit{PIF}()$, an \textit{LVR rebate parameter} $\beta \in (0,1)$, and \textit{conversion frequency} $\uptau \in \mathbb{N}$. 

We define the \textit{corresponding CFMM pool} of $\Phi$, denoted $\textit{CFMM}(\Phi)$, as the CFMM pool with reserves $(R_{x},R_{y})$ whose feasible set is described by pool invariant function $\textit{PIF}()$ and pool constant $k=\textit{PIF}(R_{x},R_{y})$. Conversely, $\Phi$ is the \textit{corresponding Diamond pool} of $\textit{CFMM}(\Phi).$ It is important to note that the mapping of $\Phi$ to $\textit{CFMM}(\Phi)$ is only used to describe the state transitions of $\Phi$, with $\textit{CFMM}(\Phi)$ changing every time the $\Phi$ pool reserves change.

Consider pool reserves $(R_{x,0},R_{y,0})$ in $\Phi$ at time $t=0$ (start of a block), and an arbitrageur wishing to move the price of $\Phi$ at time $t=1$ (end of the block) to $p_1= \frac{R_{x,1}}{R_{y,1}} \neq \frac{R_{x,0}}{R_{y,0}}$. 
In Diamond, to interact with the pool at time $t=0$, the arbitrageur must deposit some amount of collateral, $(C_x, C_y)\in \mathbb{R}_{+}^2  $. This is termed the \textit{pool unlock transaction}. 
After the pool unlock transaction, the arbitrageur can then execute arbitrarily many orders (on behalf of themselves or users) against $\Phi$, exactly as the orders would be executed in $\textit{CFMM}(\Phi)$, as long as for any intermediate reserve state $(R_{x,i},R_{y,i})$ after an order, the following holds:
\begin{equation}\label{eq:collateralRestriction}
    C_x \geq \beta(R_{x,0}-R_{x,i}) \text{ and } C_y \geq \beta(R_{y,0}-R_{y,i}).
\end{equation}

For end of block pool reserves $(R_{x,1},R_{y,1})$, WLOG let $\Upsilon_y=R_{y,1}-R_{y,0}>0$, and $\Upsilon_x=R_{x,0}-R_{x,1}>0$ (the executed orders net bought $x$ from $\Phi$, and net sold $y$ to $\Phi$). The protocol then removes $\beta \Upsilon_y$ tokens from $\Phi$, sending them to the arbitrageur, and adds $\beta \Upsilon_x$ tokens to $\Phi$, taking these tokens from $C_x$. After this, it can be seen that $ \textit{PPF}(R_{x,1}+\beta \Upsilon_x, R_{y,1}-\beta \Upsilon_y) < p_1$. To ensure the reserves correspond to a $\textit{PPF}$ equal to $p_1$, a further $\upsilon_x>0$ tokens are removed such that:
\begin{equation}\label{eq:finalReserveDef}
        \textit{PPF}(R_{x,1}+\beta \Upsilon_x - \upsilon_x, R_{y,1}-\beta \Upsilon_y)=p_1.\footnote{Achievable as a result of properties(a) and (b) of Equation \ref{eq:pricerestrictions}.} 
    \end{equation}
These $\upsilon_x$ tokens are added to the \textit{vault} of $\Phi$. Summarizing the transition from $t=0$ to $t=1$ from the arbitrageur's perspective, this is equivalent to:
\begin{enumerate}
    \item Adding $(1-\beta)\Upsilon_y$ tokens to $\Phi$ and removing $(1-\beta)\Upsilon_x$ tokens from $\Phi$.
    \item Adding $\upsilon_x>0$ tokens to the $\Phi$ vault from the $\Phi$ pool such that $\textit{PPF}(R_{x,0}-(1-\beta)\Upsilon_x - \upsilon_x, R_{y,0}+(1-\beta) \Upsilon_y)=p_1.$ Note, this is with respect to starting reserves. \footnote{If $\Upsilon_y>0$ tokens are to be removed from $\textit{CFMM}(\Phi)$ with $\Upsilon_x>0$ tokens to be added in order to achieve $p_1$, then $(1-\beta)\Upsilon_y$ tokens are removed from $\Phi$ and $(1-\beta)\Upsilon_x$ tokens are added to $\Phi$, with a further $\upsilon_y>0$ removed from $\Phi$ and added to the vault such that $\textit{PPF}(R_{x,0}+(1-\beta)\Upsilon_x,R_{y,0}-(1-\beta) \Upsilon_y-\upsilon_y)=p_1$.}
\end{enumerate} 

If only a single arbitrageur order is executed on $\Phi$ apart from the pool unlock transaction, the arbitrageur receives $\Upsilon_x$ tokens from the order, and must repay $\beta \Upsilon_x$ as a result of the pool unlock transaction. Any other sequence of orders resulting in a net $\Upsilon_x$ tokens being removed from $\Phi$ is possible, but $\beta \Upsilon_x$ tokens must always be repaid to the pool by the arbitrageur. As the arbitrageur has full control over which orders are executed, such sequences of orders must be at least as profitable for the arbitrageur as the single arbitrage order sequence. \footnote{An example of such a sequence is an arbitrage order to the external market price, followed by a sequence of order pairs, with each pair a user order, followed by an arbitrageur order back to the external market price. There are arbitrarily many other such sequences.}

\subsubsection*{Vault Rebalance.}

After the above process, let there be $(v_x,v_y) \in \mathbb{R}_{+}^2$ tokens in the vault of $\Phi$. If $v_y \MIFP_1 > v_x $, add $(v_x, \frac{v_x}{\MIFP_1}) $ tokens into $\Phi$ from the vault. Otherwise, add $(v_y \MIFP_1, v_y)$ tokens into $\Phi$ from the vault. This is a \textit{vault rebalance}.

Every $\uptau$ blocks, after the vault rebalance, the protocol converts half of the tokens still in the vault of $\Phi$ (there can only be one token type in the vault after a vault rebalance) into the other token in $\Phi$ according to one of either conversion process 1 (Section \ref{sec:Diamond:conv1}) or 2 (Section \ref{sec:Diamond:conv2}). The goal of the conversion processes is to add the Diamond vault tokens back into the Diamond liquidity pool in a ratio corresponding to the $\MIFP$, while preserving the value of the tokens to be added to the pool. 

To understand why half of the tokens are converted, assume WLOG that there are $v_x$ tokens in the vault.  Given an external market price $\MIFP$, $\frac{v_x}{2}$ tokens can be exchanged for $v_y=\frac{v_x}{2}\frac{1}{\MIFP}$ tokens, and vice versa. Both conversion processes are constructed to ensure the expected revenue of conversion is at least $v_y=\frac{v_x}{2}\frac{1}{\MIFP}$. Therefore, after conversion, there are at least $\frac{v_x}{2}$ and $v_y=\frac{v_x}{2}\frac{1}{\MIFP}$ tokens in the vault, with $\frac{\frac{v_x}{2}}{v_y}=\MIFP$. The conversion processes then add the unconverted $\frac{v_x}{2}$ and converted $v_y$ tokens back into the $\Phi$ pool, with the ratio of these tokens approximating the external market price. Importantly, these tokens have value of at least the original vault tokens $v_x$.

\subsection{Per-block Conversion vs. Future Contracts} \label{sec:Diamond:conv1}

After every arbitrage, the arbitrageur converts $\eta$ equal to half of the total tokens in the vault at the pool price $p_c$. Simultaneously, the arbitrageur sells to the pool $\eta$ future contracts in the same token denomination at price $p_c$. Given the pool buys $\eta$ future contracts at conversion price $p_c$, and the futures settle at price $p_T$, the protocol wins $\eta (p_T-p_c)$. 

These future contracts are settled every $\uptau$ blocks, with the net profit or loss being paid in both tokens, such that for a protocol settlement profit of $PnL$ measured in token $x$ and pool price $p_T$, the arbitrageur pays $(s_x,s_y)$ with $PnL=s_x+s_y p_T$ and $s_x= s_y p_T$. These contracts can be settled in one of the following (non-exhaustive) ways:
\begin{enumerate}
    \item Every $\uptau$ blocks, an auction takes place to buy the offered tokens from the arbitrageurs who converted the pool at the prices at which the conversions took place. For a particular offer, a positive bid implies the converter lost/the pool won to the futures. In this case the converter gives the tokens to the auction winner, while the pool receives the winning auction bid. A negative bid implies the converter won/the pool lost to the futures. In this case, the converter must also give the tokens to the auction winner, while the pool must pay the absolute value of the winning bid to the auction winner.
    \item Every $\uptau$ blocks, a blockchain-based frequent batch auction takes place in the swap corresponding to the pool swap. The settlement price of the frequent batch auction is used as the price at which to settle the futures.
\end{enumerate}

\subsection{Periodic Conversion Auction} \label{sec:Diamond:conv2}

Every $\uptau$ blocks, $\eta$ equal to half of the tokens in the vault are auctioned to all players in the system, with bids denominated in the other pool token (bids for $x$ tokens in the vault must be placed in $y$ tokens, and vice versa). For winning bid $b$ in token $x$ (or token $y$), the resultant vault quantities described by $(s_x=b, s_y= \eta)$ (or $(s_x=\eta, s_y= b)$) are added to the pool reserves. In this case, unlike in Section \ref{sec:Diamond:conv1}, there are no restrictions placed on $\frac{s_x}{s_y}$.

\section{\texorpdfstring{Diamond}{TEXT} Properties}\label{sec:Diamond:properties}

This section outlines the key properties of Diamond. We first prove that both conversion process have at least 0 expectancy for the protocol.

\begin{restatable}{lemma}{pbc}\label{lem:PBC}
 Converting the vault every block vs. future contracts has expectancy of at least 0 for a Diamond pool.
\end{restatable}

\begin{restatable}{lemma}{pca}\label{lem:PCA}
A periodic conversion auction has expectancy of at least 0 for a Diamond pool.
\end{restatable}

\begin{corollary}\label{lem:conversion}
    Conversion has expectancy of at least 0 for a Diamond pool.
\end{corollary}

With these results in hand, we now prove the main result of the paper. That is, the LVR of a Diamond pool is $(1-\beta)$ of the corresponding CFMM pool.

\begin{restatable}{theorem}{main}\label{thm:main}
For a CFMM pool $\textit{CFMM}(\Phi)$ with LVR of $L>0$, the LVR of $\Phi$, the corresponding pool in Diamond, has expectancy of at most $(1-\beta)L$.
\end{restatable}

\section{Implementation}\label{sec:Diamond:implementation}

We now detail an implementation of Diamond. The main focus of our implementation is ensuring user experience in a Diamond pool is not degraded compared to the corresponding CFMM pool. To this point, applying a $\beta$-discount on every Diamond pool trade is not viable. To avoid this, we only consider LVR on a per-block, and not a per-transaction basis. Given the transaction sequence, in/exclusion and priority auction capabilities of block producers, block producers can either capture the block LVR of a Diamond pool themselves, or auction this right among arbitrageurs. 

From an implementation standpoint, who captures the LVR is not important, whether it is the producer themselves, or an arbitrageur who won an auction to bundle the transactions for inclusion in Diamond. As mentioned already, we assume these are the same entity, and as such it is the arbitrageur who must repay the LVR of a block. To enforce this, for a Diamond pool, we check the pool state in the first pool transaction each block and take escrow from the arbitrageur. This escrow is be used in part to pay the realized LVR of the block back to the pool. The first pool transaction also returns the collateral of the previous arbitrageur, minus the realized LVR (computable from the difference between the current pool state and the pool state at the beginning of the previous block). To ensure the collateral covers the realized LVR, each proceeding pool transaction verifies that the LVR implied by the pool state as a result of the transaction can be repaid by the deposited collateral. 

We can reduce these collateral restrictions by allowing the arbitrageur to bundle transactions based on a \textit{coincidence-of-wants} (CoWs) (matching buy and sell orders, as is done in CoWSwap \cite{CoWSwapWebsite}). This can effectively reduce the required collateral of the arbitrageur to 0. Given the assumed oversight capabilities of arbitrageurs is the same as that of block producers, we do not see collateral lock-up intra-block as a restriction, although solutions like CoWs are viable alternatives.

Our implementation is based on the following two assumptions:
\begin{enumerate}
    \item An arbitrageur always sets the final state of a pool to the state which maximizes the LVR. 
    \item The block producer realizes net profits of at least the LVR corresponding to the final state of the pool, either as the arbitrageur themselves, or by auctioning the right to arbitrage amongst a set of competing arbitrageurs. 
\end{enumerate}

If the final price of the block is not the price maximizing LVR, the arbitrageur has ignored an arbitrage opportunity. The arbitrageur can always ignore non-arbitrageur transactions to realize the LVR, therefore, any additional included transactions must result in greater or equal utility for the arbitrageur than the LVR.

\subsection{Core Protocol}

The first transaction interacting with a Diamond pool $\pool$ in every block is the \textit{pool unlock transaction}, which deposits some collateral, $(C_x, C_y)\in \mathbb{R}_{+}^2  $. Only one pool unlock transaction is executed per pool per block. Every proceeding user order interacting with $\pool$ in the block first verifies that the implied pool move stays within the bounds of Equation \ref{eq:collateralRestriction}. Non pool-unlock transactions are executed as they would be in the corresponding CFMM pool $\textit{CFMM}(\pool)$ (without a $\beta$ discount on the amount of tokens that can be removed). These transactions are executed at prices implied by pushing along the $\textit{CFMM}(\pool)$ curve from the previous state, and as such, the ordering of transactions intra-block affects the execution price. If a Diamond transaction implies a move outside of the collateral bounds, it is not executed. 

The next time a pool unlock transaction is submitted (in a proceeding block), given the final price of the preceding block was $p_1$, the actual amount of token $x$ or $y$ required to be added to the pool and vault (the $\beta \Upsilon $ and $\upsilon$ of the required token, as derived earlier in the section) is taken from the deposited escrow, with the remainder returned to the arbitrageur who deposited those tokens.

\begin{remark}
  Setting the LVR rebate parameter too high can result in protocol censorship and/or liveness issues as certain arbitrageurs may not be equipped to frictionlessly arbitrage, and as such, repay the implied LVR to the protocol. To counteract this, the LVR rebate parameter should be reduced every block in which no transactions take place. As arbitrageurs are competing through the block producers to extract LVR from the pool, the LVR rebate parameter will eventually become low enough for block producers to include Diamond transactions. After transactions have been executed, the LVR rebate parameter should be reset to its initial value. Rigorous testing of initial values and decay curves are required for any choice of rebate parameter. 
\end{remark}

\subsection{Conversion Protocols}

The described implementations in this section assume the existence of a decentralized on-chain auction.\footnote{First-price sealed-bid auctions can be implemented using a commit-reveal protocol. An example of such a protocol involves bidders hashing bids, committing these to the blockchain along with an over-collaterlization of the bid, with bids revealed when all bids have been committed.} 

\subsubsection{Per-block Conversion vs. Futures}

Given per-block conversion (Section \ref{sec:Diamond:conv1}), further deposits from the arbitrageur are required to cover the token requirements of the conversion and collateralizing the futures. The conversions for a pool $\pool$ resulting from transactions in a block take place in the next block a pool unlock transaction for $\pool$ is called. Given a maximum expected percentage move over $\uptau$ blocks of $\sigma_T$, and a conversion of $\lambda_y$ tokens at price $p$, the arbitrageur collateral must be in quantities $\pi_x$ and $\pi_y$ such that if the arbitrageur is long the futures:
\begin{equation}
    1. \ \pi_x+\pi_y\frac{p}{1+\sigma_T}\geq\lambda_y (p-\frac{p}{1+\sigma_T}), \  \text{and} \ 2. \ 
    \frac{\pi_x}{\pi_y}=\frac{p}{1+\sigma_T}. \ \ \ \ 
\end{equation}
If the arbitrageur is short the futures it must be that:
\begin{equation}
    1. \ \pi_x+\pi_y p(1+\sigma_T)\geq \lambda_y p\sigma_T, \ \ \ \ \ \ \ \ \ \ \ \ \ \text{and} \  2. \
    \frac{\pi_x}{\pi_y}=p(1+\sigma_T).
\end{equation}
The first requirement in both statements is for the arbitrageur's collateral to be worth more than the maximum expected loss. The second requirement states the collateral must be in the ratio of the pool for the maximum expected loss (which also ensures it is in the ratio of the pool for any other loss less than the maximum expected loss). This second requirement ensures the collateral can be added back into the pool when the futures are settled.

At settlement, if the futures settle in-the-money for the arbitrageur, tokens are removed  from the pool in the ratio specified by the settlement price with total value equal to the loss incurred by the pool, and paid to the arbitrageur. If the futures settle out-of-the-money, tokens are added to the pool from the arbitrageur's collateral in the ratio specified by the settlement price with total value equal to the loss incurred by the arbitrageur. The remaining collateral is returned to the arbitrageur. The pool constant is adjusted to reflect the new balances.

\begin{remark}
  As converting the vault does not affect pool availability, the auctions for converting the vault can be run sufficiently slowly so as to eliminate the risk of block producer censorship of the auction. We choose to not remove tokens from the pool to collateralize the futures as this reduces the available liquidity within the pool, which we see as an unnecessary reduction in benefit to users (which would likely translate to lower transaction fee revenue for the pool). For high volatility token pairs, $\uptau$ should be chosen sufficiently small so as to not to risk pool liquidation. 

  If Diamond with conversion versus futures is run on a blockchain where the block producer is able to produce multiple blocks consecutively, this can have an adverse effect on incentives. Every time the vault is converted and tokens are re-added to the pool, the liquidity of the pool increases. A block producer with control over multiple blocks can move the pool price some of the way towards the maximal LVR price, convert the vault tokens (which has 0 expectancy from Lemma \ref{lem:PBC}), increase the liquidity of the pool, then move the pool towards the maximal LVR price again in the proceeding block. This process results in a slight increase in value being extracted from the pool in expectancy compared to moving the pool price immediately to the price corresponding to maximal LVR. Although the effect on incentives is small, re-adding tokens from a conversion slowly/keeping the pool constant unchanged mitigates/removes this benefit for such block producers.
\end{remark}

\subsubsection{Periodic Conversion Auction}

Every $\uptau$ blocks, $\eta$ equal to half the tokens in the vault are auctioned off, with bids denominated in the other token. The winning bidder receives these $\eta$ tokens. The winning bid, and the remaining $\eta$ tokens in the vault, are re-added to the pool.

\begin{figure}[t]
    \centering
    \begin{minipage}{0.45\textwidth}
        \centering
        \includegraphics[scale=0.4]{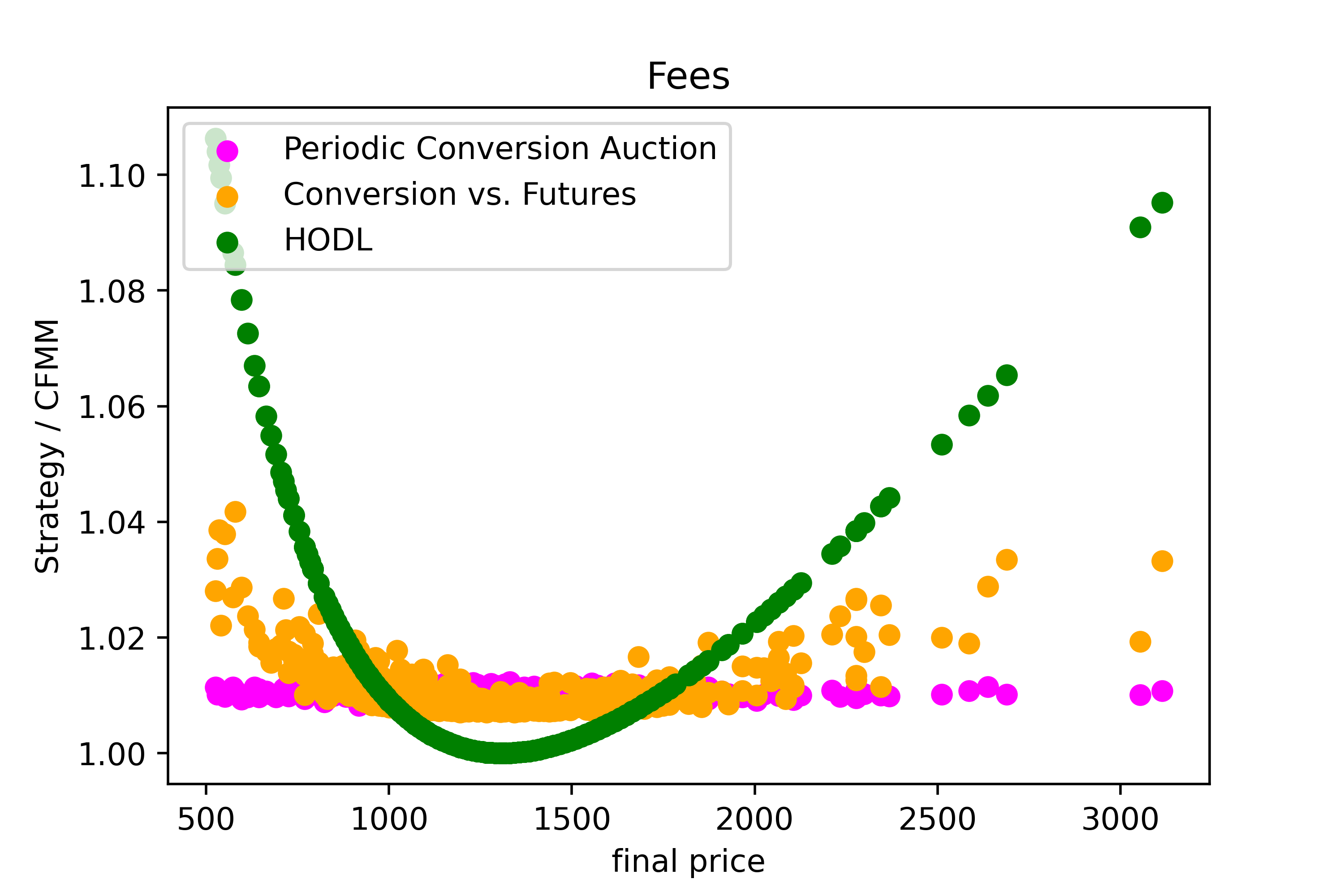}
        \caption{ }
        \label{fig:StrategyComparison}
    \end{minipage}\hfill
    \begin{minipage}{0.45\textwidth}
        \includegraphics[scale=0.4]{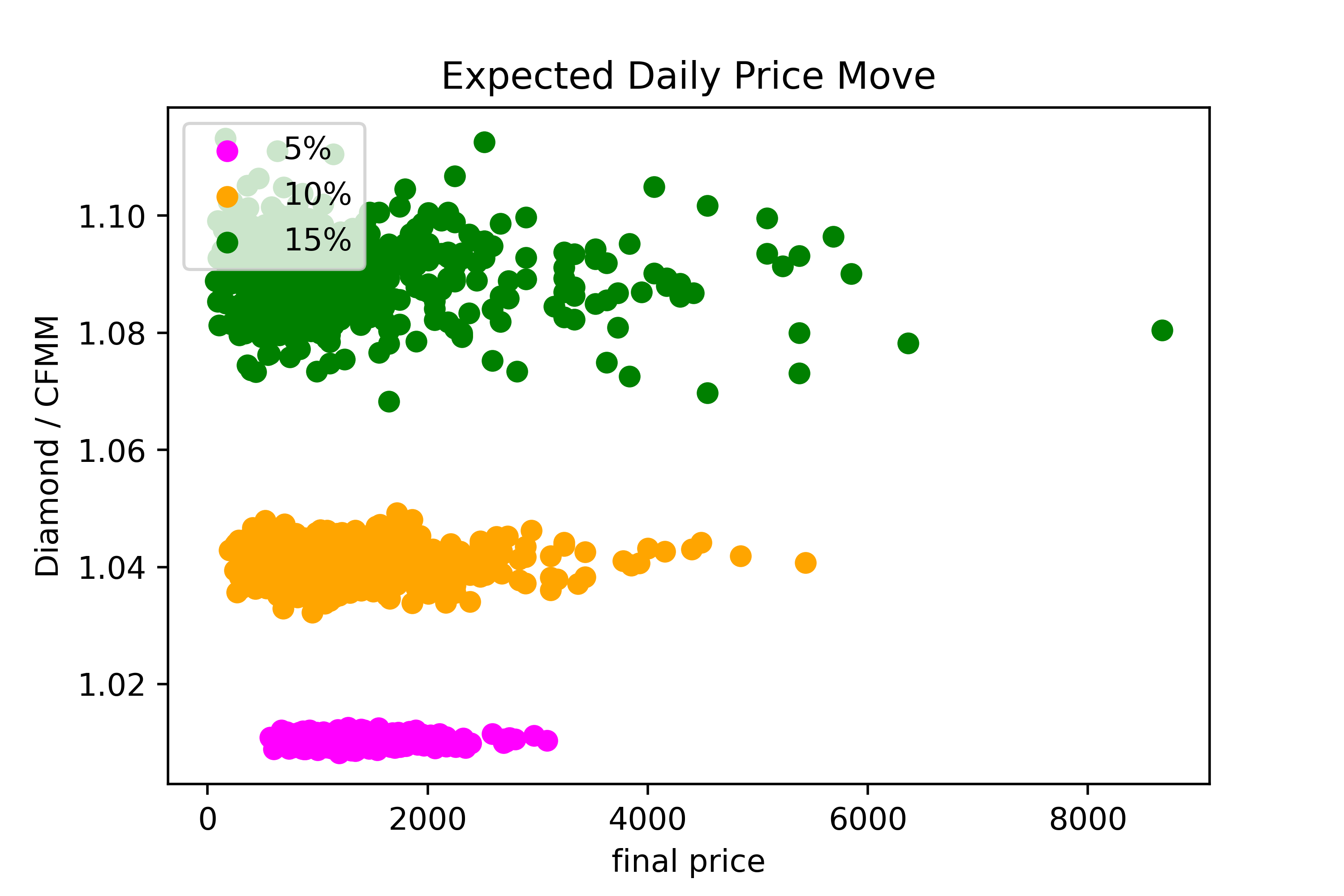}
        \caption{ }
        \label{fig:VolatilityComparison}
    \end{minipage}
\end{figure}

\section{Experimental Analysis}\label{sec:Diamond:sims}

This section presents the results of several experiments, which can be reproduced using the following public repository \cite{DiamondprotocolGithub}. The results provide further evidence of the performance potential of a Diamond pool versus various benchmarks. These experiments isolate the effect that different fees, conversion frequencies, daily price moves, LVR rebate parameters, and days in operation have on a Diamond pool. Each graph represents a series of random-walk simulations which were run, unless otherwise stated, with base parameters of:
\vspace{-0.11cm}
\begin{itemize}
    \item LVR rebate parameter: $0.95$.
    \item Average daily price move: $5\%$.
    \item Conversion frequency: Once per day.
    \item Blocks per day: $10$.
    \item Days per simulation: $365$.
    \item Number of simulations per variable: $500$.
\end{itemize}
\vspace{-0.5cm}
\subsubsection*{Parameter Intuition.} For a Diamond pool to be deployed, we expect the existence of at least one tradeable and liquid external market price. As such, many competing arbitrageurs should exist, keeping the LVR parameter close to 1. $5\%$ is a typical daily move for our chosen token pair. Given a daily move of $5\%$, the number of blocks per day is not important, as the per block expected moves can be adjusted given the daily expected move. Given a simulator constraint of 5,000 moves per simulation, we chose $10$ blocks per day for a year, as opposed to simulating Ethereum over 5,000 blocks (less than 1 day's worth of blocks), as the benefits of Diamond are more visible over a year than a day. 
\vspace{0.1cm}

Each graph plots the final value of the Diamond Periodic Conversion Auction pool (unless otherwise stated) relative to the final value of the corresponding Uniswap V2 pool. The starting reserve values are $\$100m$ USDC and $76,336$ ETH, for an ETH price of $\$1,310$, the approximate price and pool size of the Uniswap ETH/USDC pool at the time of simulation \cite{UniswapWebsite}. 
Figure \ref{fig:StrategyComparison} compares four strategies over the same random walks. Periodic Conversion Auction and Conversion vs. Futures replicate the Diamond protocol given the respective conversion strategies (see Section \ref{sec;protocol}).    HODL (Hold-On-for-Dear-Life), measures the performance of holding the starting reserves until the end of the simulation. The final pool value of these three strategies are then taken as a fraction of the corresponding CFMM pool following that same random walk. Immediately we can see all three of these strategies outperform the CFMM strategy in all simulations (as a fraction of the CFMM pool value, all other strategies are greater than 1), except at the initial price of 1310, where HODL and CFMM are equal, as expected. 

\begin{figure}[t]
    \centering
    \begin{minipage}{0.45\textwidth}
        \centering
        \includegraphics[scale=0.4]{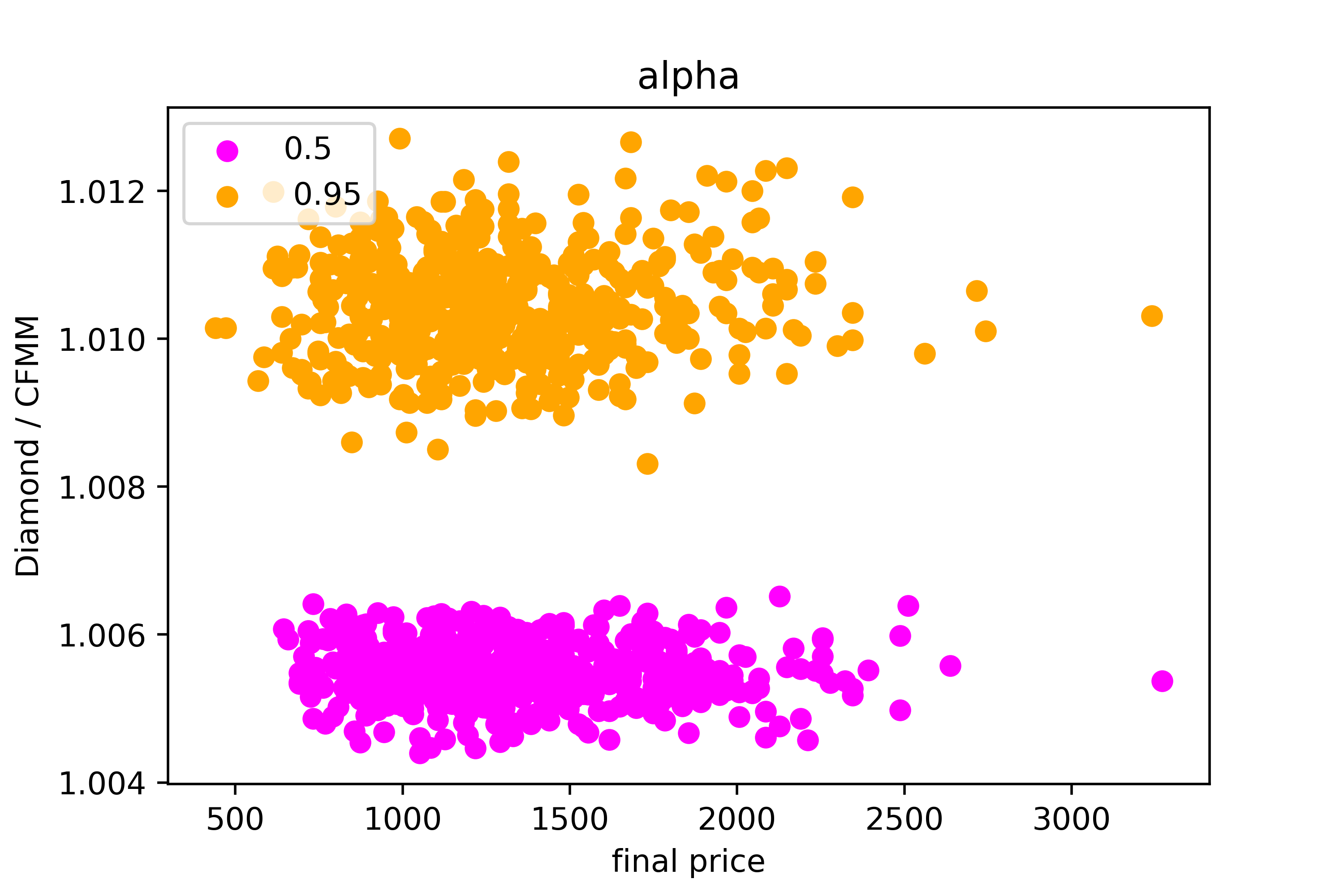}
        \caption{}
        \label{fig:AlphaComparison}
    \end{minipage}\hfill
    \begin{minipage}{0.45\textwidth}
        \includegraphics[scale=0.4]{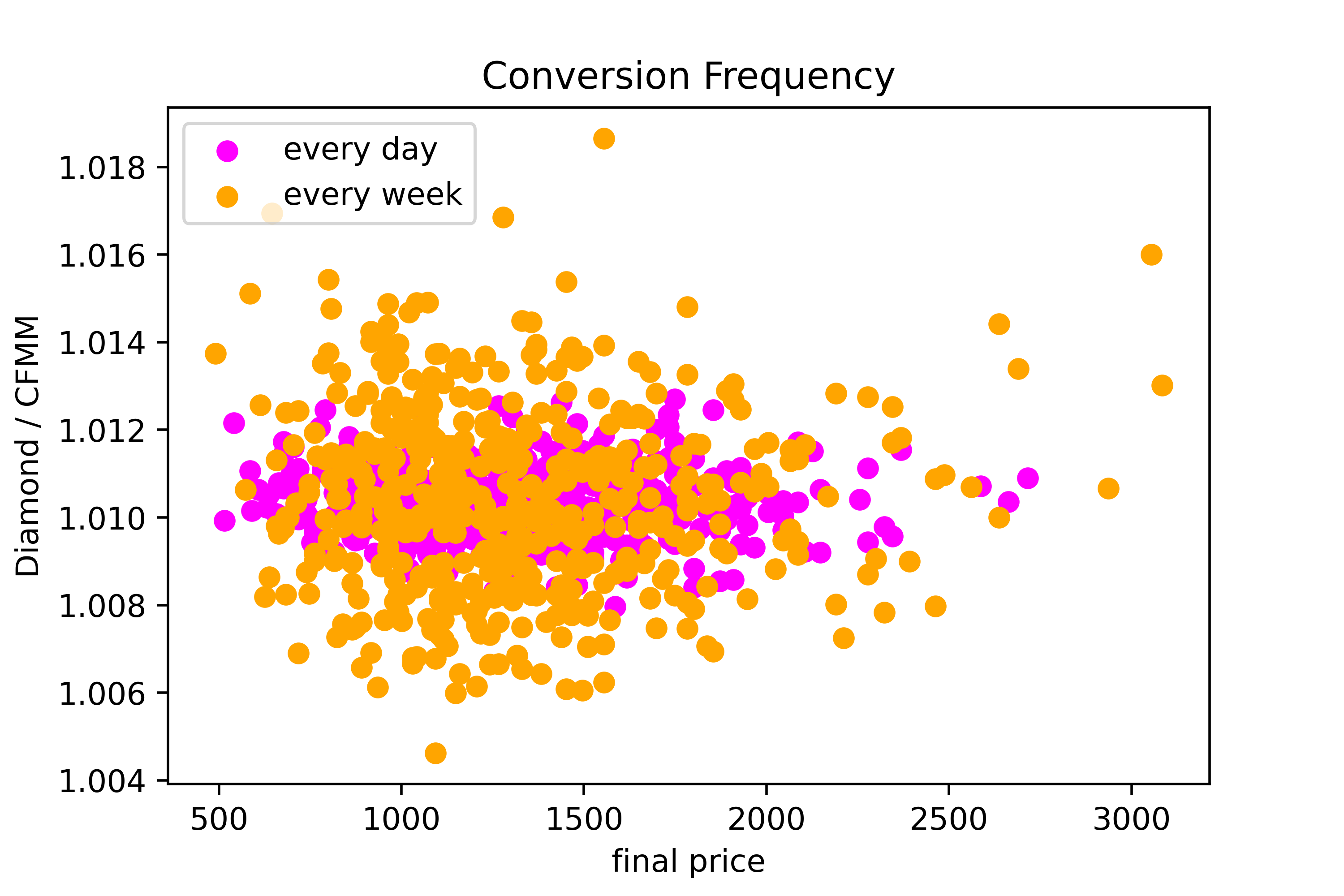}
        \caption{}
        \label{fig:ConversionFrequencyComparison}
    \end{minipage}
\end{figure}

The Diamond pools outperform HODL in a range around the starting price, as Diamond pools initially retain the tokens increasing in value (selling them eventually), which performs better than HODL when the price reverts. HODL performs better in tail scenarios as all other protocols consistently sell the token increasing in value on these paths. Note Periodic Conversion slightly outperforms Conversion vs. Futures when finishing close to the initial price, while slightly underperforming at the tails. This is because of the futures exposure. Although these futures have no expectancy for the protocol, they increase the variance of the Conversion vs. Futures strategy, outperforming when price changes have momentum, while underperforming when price changes revert.
 
Figure \ref{fig:VolatilityComparison} identifies a positive relationship between the volatility of the price and the out-performance of the Diamond pool over its corresponding CFMM pool. This is in line with the results of \cite{LVRRoughgarden} where it is proved LVR grows quadratically in volatility. Figure \ref{fig:AlphaComparison} demonstrates that, as expected, a higher LVR rebate parameter $\beta$ retains more value for the Diamond pool.

Figure \ref{fig:ConversionFrequencyComparison} shows that higher conversion frequency (1 day) has less variance for the pool value (in this experiment once per day conversion has mean 1.011234 and standard deviation 0.000776 while once per week conversion has mean 1.011210 and standard deviation 0.002233). This highlights an important trade-off for protocol deployment and LPs. Although lower variance corresponding to more frequent conversion auctions is desirable, more frequent auctions may centralize the players participating in the auctions due to technology requirements. This would weaken the competition guarantees needed to ensure that the auction settles at the true price in expectancy. 

\begin{figure}[t]
    \centering
    \begin{minipage}{0.45\textwidth}
        \centering
        \includegraphics[scale=0.4]{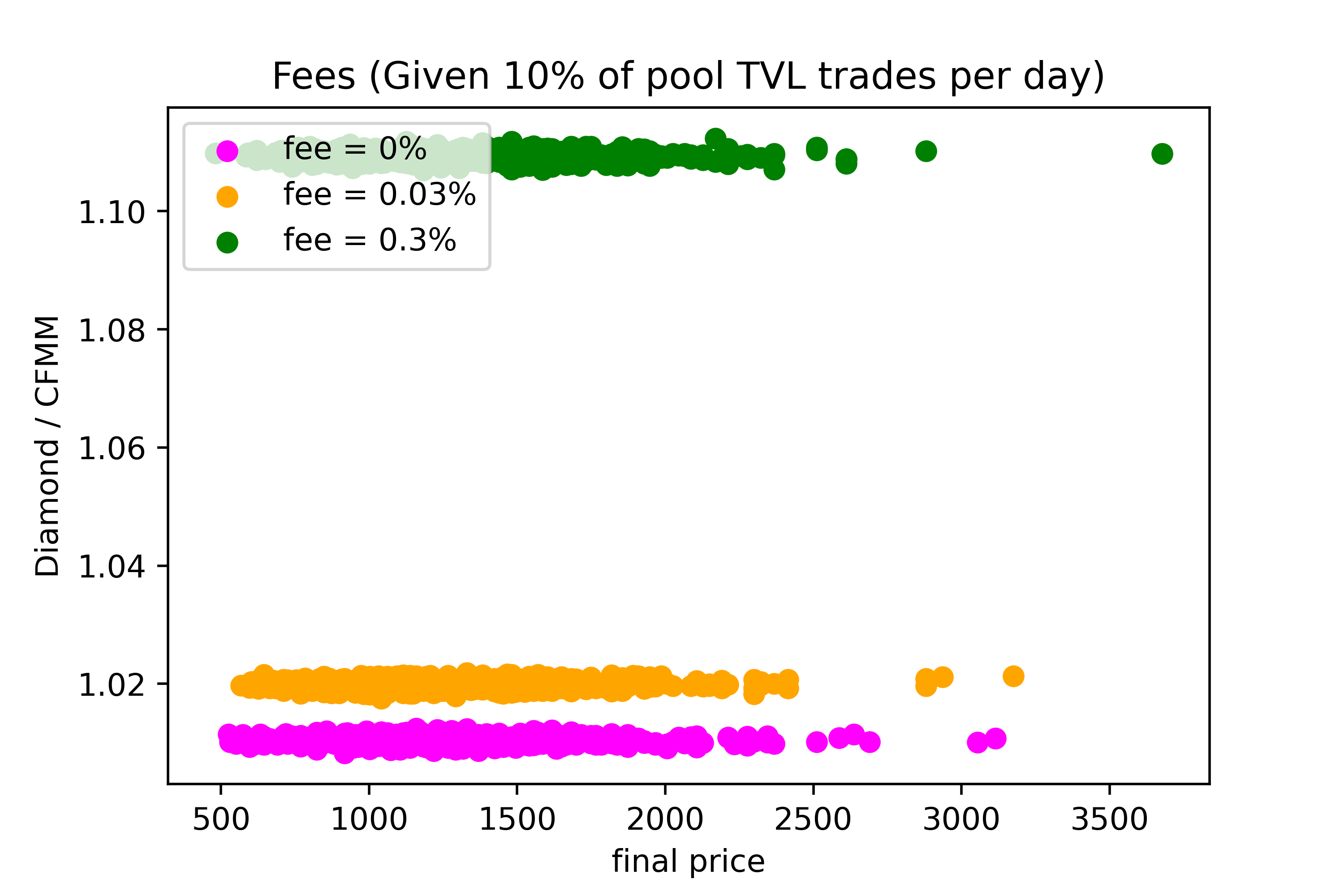}
        \caption{ }
        \label{fig:TransactionFeeComparison}
    \end{minipage}\hfill
    \begin{minipage}{0.45\textwidth}
        \includegraphics[scale=0.4]{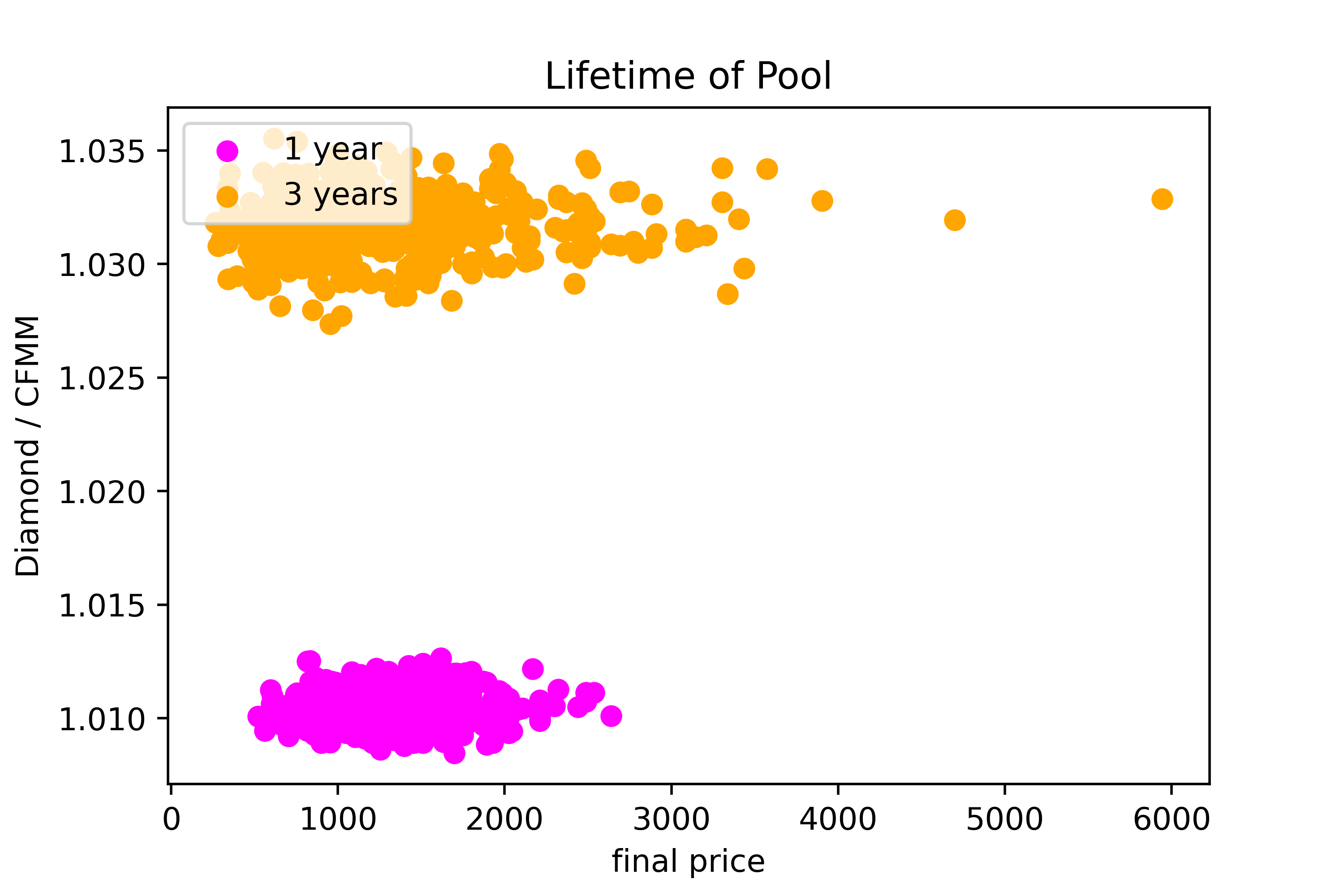}
        \caption{ }
        \label{fig:PoolLifetimeComparison}
    \end{minipage}
\end{figure}

Figure \ref{fig:TransactionFeeComparison} compares Diamond to the CFMM pool under the specified fee structures (data-points corresponding to a particular fee apply the fee to both the Uniswap pool and the Diamond pool) assuming $10\%$ of the total value locked in each pool trades daily. The compounding effect of Diamond's LVR rebates with the fee income every block result in a significant out-performance of the Diamond protocol as fees increase. This observation implies that given the LVR protection provided by Diamond, protocol fees can be reduced significantly for users, providing a further catalyst for a DeFi revival. 
Figure \ref{fig:PoolLifetimeComparison} demonstrates that the longer Diamond is run, the greater the out-performance of the Diamond pool versus its corresponding CFMM pool.

\section{Conclusion}\label{sec:Diamond:conclusion}

We present Diamond, an AMM protocol which provably protects against LVR.  The described implementation of Diamond stands as a generic template to address LVR in any CFMM. The experimental results of Section \ref{sec:Diamond:sims} provide strong evidence in support of the LVR protection of Diamond, complementing the formal results of Section \ref{sec:Diamond:properties}. It is likely that block producers will be required to charge certain users more transaction fees to participate in Diamond pools to compensate for this LVR rebate, with informed users being charged more for block inclusion than uninformed users. As some or all of these proceeds are paid to the pool with these proceeds coming from informed users, we see this as a desirable outcome. 

\section{Acknowledgements}

We thank the reviewers for their detailed and insightful reviews, as well as Stefanos Leonardos for his guidance in preparing this camera-ready version. 
This paper is part of a project that has received funding from the European Union's Horizon 2020 research and innovation programme under grant agreement number 814284, and is supported by the AEI-PID2021-128521OB-I00 grant of the Spanish Ministry of Science and Innovation.

\addcontentsline{toc}{section}{Bibliography}
\bibliographystyle{splncs04}
\bibliography{references}

\appendix
\section{Proofs}\label{sec:proofs}

\auction*

\begin{proof} By construction, we have that the support of $X_i$ is lower bounded by $z\MIFP$. Therefore, in a second-price auction, in equilibrium, each player will bid at least, $z\MIFP$. Using the revenue equivalence theorem \cite{krishna2009auction}, we deduce that the revenue of the seller is at least $z\MIFP$ obtaining the result.
\end{proof}

\pbc*

\begin{proof}
    Consider a conversion of $\eta$ tokens which takes place at time 0. Let the conversion be done at some price $p_c$, while the external market price is $\MIFP_0$. WLOG let the protocol be selling $\eta$ $y$ tokens in the conversion, and as such, buying $\eta$ $y$ token futures at price $p_c$. 
    The token sells have expectancy $\eta (p_c-\MIFP_0)$. For the strategy to have at least 0 expectancy, we need the futures settlement to have expectancy of at least $\eta (\MIFP_0-p_c)$. In Section \ref{sec:Diamond:conv1}, two versions of this strategy were outlined. We consider both here. In both sub-proofs, we use the assumption that the risk-free rate is 0, which coupled with our martingale assumption for $\MIFP$ means the external market price at time $t$ is such that $\mathbb{E}(\MIFP_t)=\MIFP_0$. We now consider the two options for settling futures outlined in Section \ref{sec:Diamond:conv1}
    
    \textbf{Option 1: Settle futures by auctioning tokens at the original converted price.} The arbitrageur who converted tokens for the pool at price $p_c$ must auction off the tokens at price $p_c$. Let the auction happen at time $t$, with external market price at that time of $\MIFP_t$. Notice that what is actually being sold is the right, and obligation, to buy $\eta$ tokens at price $p_c$. This has value $\eta(\MIFP_t-p_c)$, which can be negative. As negative bids are paid to the auction winner by the protocol, and positive bids are paid to the protocol, we are able to apply Lemma \ref{lem:incentiveCompatibleAuctions}. As such, the winning bid is at least $\eta(\MIFP_t-p_c)$, which has expectancy of at least
    \begin{equation}
        \mathbb{E}(\eta(\MIFP_t-p_c))=\eta(\mathbb{E}(\MIFP_t)-p_c)=\eta(\MIFP_0-p_c).
    \end{equation}
    Thus the expectancy of owning the future for the protocol is at least $\eta (\MIFP_0-p_c)$, as required. 
    
    \textbf{Option 2: Settle futures using frequent batch auction settlement price.} For a swap with external market price $\MIFP_t$ at time $t$, a batch auction in this swap settles at $\MIFP_t$ in expectancy (Theorem 5.1 in \cite{FairTraDEXMcMenamin}). Thus the futures owned by the protocol have expectancy
    \begin{equation}
        \mathbb{E}(\eta(\MIFP_t-p_c))=\eta(\mathbb{E}(\MIFP_t)-p_c)=\eta(\MIFP_0-p_c).
    \end{equation}
\end{proof}

\pca*

\begin{proof}
    Consider a Diamond pool $\Phi$ with vault containing $2\eta$ tokens. WLOG let these be of token $y$. Therefore the pool must sell $\eta$ tokens at the external market price to balance the vault. Let the conversion auction accept bids at time $t$, at which point the external market price is $\MIFP_t$. For the auction to have expectancy of at least 0, we require the winning bid to be at least $\eta \MIFP_t$. The result follows from Lemma \ref{lem:incentiveCompatibleAuctions}. 
\end{proof}

\main*

\begin{proof}
    To see this, we first know that for $CFMM(\Phi)$ at time $t$ with reserves $(R_{x,t},R_{y,t})$, LVR corresponds to the optimal solution $(R^*_{x,t+1},R^*_{y,t+1}) $ with external market price $\MIFP_{t+1}$ which maximizes:
    \begin{equation}
        (R_{x,t+1}-R_{x,t})+ (R_{y,t+1}-R_{y,t})\MIFP_{t+1}.
    \end{equation}
    Let this quantity be
    \begin{equation}\label{eq:normalLVR}
        L=(R^*_{x,t+1}-R_{x,t})+ (R^*_{y,t+1}-R_{y,t})\MIFP_{t+1}.
    \end{equation}
    In Diamond, a player trying to move the reserves of $\Phi$ to $(R'_{x,t+1},R'_{y,t+1})$ only receives $(1-\beta)(R'_{x,t+1}-R_{x,t})$ while giving $(1-\beta) (R'_{y,t+1}-R_{y,t})$ to $\Phi$. Thus, an arbitrageur wants to find the values of $(R'_{x,t+1},R'_{y,t+1})$ that maximize:
    \begin{equation}\label{eq:thmLVR}
        (1-\beta)(R'_{x,t+1}-R_{x,t})+ (1-\beta)(R'_{y,t+1}-R_{y,t})\MIFP_{t+1} + \mathbb{E}(\text{conversion}).
    \end{equation}
    where $\mathbb{E}(\text{conversion})$ is the per-block amortized expectancy of the conversion operation for the arbitrageurs. From Lemma \ref{lem:conversion}, we know $\mathbb{E}(\text{conversion})\geq 0$ for $\Phi$. This implies the arbitrageur's max gain is less than: 
    \begin{equation}\label{eq:thmLVR2}
        (1-\beta)(R'_{x,t+1}-R_{x,t})+ (1-\beta)(R'_{y,t+1}-R_{y,t})\MIFP_{t+1},
    \end{equation}
    for the $(R'_{x,t+1},R'_{y,t+1})$ maximizing Equation \ref{eq:thmLVR}. 
    From Equation \ref{eq:normalLVR}, we know this has a maximum at $(R'_{x,t+1},R'_{y,t+1})= (R^*_{x,t+1},R^*_{y,t+1})$. Therefore, the LVR of $\Phi$ is at most:
    \begin{equation}
        (1-\beta)((R^*_{x,t+1}-R_{x,t})+ (R^*_{y,t+1}-R_{y,t})\MIFP_{t+1})=(1-\beta)L.
    \end{equation}
\end{proof}

\end{document}